\documentclass[aps,pra,reprint,superscriptaddress,nofootinbib,longbibliography,floatfix]{revtex4-2}

\usepackage{amsmath,amssymb,mathtools,bm}
\usepackage{braket}
\usepackage{amsthm}
\usepackage{graphicx}
\usepackage{tikz}
\usetikzlibrary{positioning}
\usepackage{microtype}
\usepackage{hyperref}
\usepackage{booktabs}
\usepackage{xcolor}

\hypersetup{colorlinks=true,citecolor=blue,linkcolor=blue,urlcolor=blue}

\newcommand{\sprof}{\mathfrak{s}}
\newcommand{\Zd}{\mathbb{Z}_{d}}
\newcommand{\Vn}{V^n}
\newcommand{\id}{\mathbb{I}}
\newcommand{\cH}{\mathcal{H}}
\newcommand{\cM}{\mathcal{M}}
\newcommand{\cN}{\mathcal{N}}

\newcommand{\Tr}{\operatorname{tr}}
\newcommand{\supp}{\operatorname{supp}}

\newcommand{\hbin}{h_2}
\newcommand{\ketbra}[2]{|#1\rangle\!\langle #2|}
\newcommand{\proj}[1]{|#1\rangle\!\langle #1|}
\newcommand{\dd}{\mathrm{d}}
\newtheorem{theorem}{Theorem}
\newtheorem{lemma}[theorem]{Lemma}
\newtheorem{proposition}[theorem]{Proposition}
\newtheorem{corollary}[theorem]{Corollary}
\newtheorem{problem}[theorem]{Problem}

\begin{document}

\title{Mean-State Entropy Hierarchies and Classical Communication through Quantum Convolutions}

\author{Chunhe Xiong}
\email{xiongchunhe@csu.edu.cn}
\affiliation{School of Mathematics and Statistics, Central South University, Changsha 410083, China}

\author{Sunho Kim}
\affiliation{School of Mathematical Sciences, Harbin Engineering University, Harbin 150001, China}

\author{Qing-Hua Zhang}
\affiliation{School of Mathematics and Statistics, Changsha University of Science and Technology, Changsha 410114, China}

\author{Shao-Ming Fei}
\email{feishm@cnu.edu.cn}
\affiliation{School of Mathematical Sciences, Capital Normal University, Beijing 100048, China}

\author{Junde Wu}
\email{wjd@zju.edu.cn}
\affiliation{School of Mathematical Sciences, Zhejiang University, Hangzhou 310027, China}

\begin{abstract}

Quantum convolution provides a discrete-variable analogue of classical convolution, with the mean state capturing the stabilizer structure preserved under repeated convolution. We establish a finite-step entropy hierarchy generated by compatible stabilizer dephasings. Along every compatible isotropic flag, the entropy increases toward the mean-state entropy ceiling, while the relative-entropy distance to the mean state decomposes exactly into successive coherence losses and a terminal classical nonuniformity. Optimizing over compatible subspaces yields an intrinsic entropy profile of the state. For quantum convolutional channels, Weyl covariance reduces the one-shot classical communication problem to minimal output entropy. A spectral-transfer argument shows that suitable stabilizer inputs reproduce stabilizer-measurement distributions of the environment as channel-output spectra. This gives a computable Holevo lower bound over all complete stabilizer measurements; its compatible restriction is characterized by the entropy hierarchy and refines the previous mean-state bound of Bu, Gu, and Jaffe.
The bound is exact for stabilizer-diagonal environments, for which the Holevo capacity is strongly additive, and yields a single-letter formula for a nonstabilizer qutrit family. The same family also exhibits a coexistence region with simultaneously positive classical and quantum communication rates.


\end{abstract}

\maketitle

\section{Introduction}

In classical probability, convolution is tied to two basic structures: multiplication of Fourier transforms and the emergence of Gaussian equilibrium under iteration. Quantum convolution imports this picture into a finite-dimensional phase space. Weyl characteristic functions replace classical Fourier transforms, stabilizer structure replaces the role of linear constraints, and repeated convolution leads to a distinguished equilibrium object known as the mean state \cite{Wootters1987ap,Gottesman1997phd,Gross2006jmp,King2003tit,Bu2025prl,Bu2023discrete,Bu2023pnas}. This makes quantum convolution a useful setting in which ideas from stabilizer theory, entropy inequalities, and quantum communication can be studied within the same algebraic framework.

The mean state was introduced in the study of discrete quantum central-limit behavior \cite{Bu2023pnas,Bu2023discrete} and later appeared in related convolutional studies of stabilizer and magic structure \cite{Bu2025tit}.
It records precisely the Weyl modes whose characteristic coefficients have unit modulus and therefore survive the convolutional dynamics. In this sense, it plays a role analogous to a discrete Gaussian endpoint. Previous work has also clarified its extremal character: once the exact stabilizer constraints are fixed, the mean state is the maximally mixed state on the corresponding stabilizer sector and hence the entropy-maximizing completion of those constraints. This endpoint picture is conceptually important, but it leaves open a more local question. Before the state reaches the mean-state entropy ceiling, how is its information lost under finite stabilizer coarse-grainings?

We address this question by introducing a finite-step entropy hierarchy associated with the mean state. Starting from the exact stabilizer skeleton of a state, we consider isotropic flags that add compatible stabilizer observables one at a time. The corresponding stabilizer dephasings produce a monotone sequence of entropies below the mean-state entropy ceiling. More importantly, the increase at each step is not merely monotonic: it is exactly a relative-entropy loss. As a result, the total relative-entropy distance to the mean state decomposes into successive losses of stabilizer-resolved coherence and a final classical nonuniformity term. Optimizing over compatible stabilizer subspaces gives an intrinsic entropy profile of the state. Unlike the mean-state entropy, which describes only the final equilibrium value, the entropy profile resolves how stabilizer information is progressively lost at different measurement resolutions. Thus the mean state is not only an asymptotic endpoint, but also organizes a finite-resolution entropy geometry around the stabilizer skeleton.

We then apply this state-space structure to communication through quantum convolutional channels. Classical communication over a quantum channel is governed by the Holevo information, whose regularization is generally unavoidable because additivity fails in general. The additivity problem is closely connected with the minimum output entropy and output purity of quantum channels, and has been extensively studied through both positive results for symmetric channels and counterexamples to multiplicativity conjectures
\cite{Holevo1973pit,Schumacher1997pra,Holevo1998tit,Werner2002jmp,King2001tit,Hastings2009np,Hayden2008cmp,Fukuda2010cmp}. Nevertheless, symmetric channels and entanglement-breaking channels provide important cases where one-shot entropy formulae or additivity theorems are available \cite{King2003tit,Shor2002jmp}. Quantum convolutional channels form a natural structured family in this direction. Their Weyl covariance reduces the one-shot Holevo problem, under the usual positivity condition on the convolution matrix, to a minimal-output-entropy problem.

The communication application has two layers. First, a spectral-transfer argument shows that every complete stabilizer measurement of the environment can be reproduced as the output spectrum of a suitable stabilizer channel input. This yields an optimized stabilizer-dephasing lower bound over all complete stabilizer measurements. Second, restricting this bound to stabilizer measurements compatible with the exact stabilizer skeleton connects it to the entropy hierarchy and gives the canonical lower bound $n\log d-\mathfrak{s}_n(\sigma)$, which refines the earlier mean-state bound.

We also identify cases where this lower bound is exact. For stabilizer-diagonal environments, the convolutional channel is entanglement-breaking, and the Holevo capacity is strongly additive. This gives an exact regularized capacity formula in terms of the stabilizer-dephasing entropy. Beyond stabilizer environments, we analyze a one-parameter qutrit family for which the single-letter stabilizer-dephasing formula can be proved explicitly. The same family also admits positive coherent-information lower bounds on open parameter regions. Hence non-stabilizer environments can support simultaneous classical and quantum communication through quantum convolutional channels; this coexistence is consistent with, but conceptually distinct from, recent studies of magic resources in quantum communication, where non-stabilizerness has been shown to influence quantum capacity in related discrete beam-splitter settings \cite{Bu2025prl}.


The paper is organized as follows. Section~II fixes notation for Weyl systems, stabilizer dephasings, mean states, and quantum convolutional channels. Section~III proves Weyl covariance and the minimal-output-entropy reduction. Section~IV develops the mean-state entropy hierarchy and its relative-entropy decomposition. Section V proves the stabilizer spectral-transfer Holevo bound. Section VI gives exact, robust, and single-letter formulae, including stabilizer-diagonal environments and the qutrit nonstabilizer family. Section VII studies the coexistence of classical and quantum communication.


\section{Framework}

All logarithms are base two. Throughout, \(d\) is an odd prime and \(n\) is a positive integer.

\subsection{Weyl operators and characteristic functions}

Let \(\cH\simeq(\mathbb{C}^d)^{\otimes n}\), let \(\omega_d=e^{2\pi i/d}\), and define on one qudit
\begin{equation*}
 X\ket{k}=\ket{k+1},\qquad Z\ket{k}=\omega_d^k\ket{k},
\end{equation*}
with arithmetic modulo \(d\). For \(z=(p,q)\in \Vn:=\Zd^n\times\Zd^n\), set
\begin{equation*}
 w(p,q)=\omega_d^{-2^{-1}p\cdot q}Z^pX^q,
\end{equation*}
where \(2^{-1}\) is the inverse of \(2\) in \(\Zd\). The symplectic form is
\begin{equation*}
 \langle(p,q),(p',q')\rangle_s=p\cdot q'-p'\cdot q.
\end{equation*}
The Weyl operators satisfy
\begin{equation*}
 w(z')w(z)w(z')^\dagger=\omega_d^{\langle z',z\rangle_s}w(z).
\end{equation*}
The characteristic function of a state \(\rho\) is
\begin{equation}
 \Xi_\rho(z)=\Tr[\rho w(-z)],
 \rho=d^{-n}\sum_{z\in\Vn}\Xi_\rho(z)w(z).
 \label{eq:weyl-expansion}
\end{equation}

A subspace \(\cN\subseteq\Vn\) is isotropic when \(\langle u,v\rangle_s=0\) for all \(u,v\in N\). We use \(N\) for a general isotropic subspace, and reserve calligraphic letters such as \(\mathcal M\) for maximal isotropic subspaces. 
 If \(\dim N=k\), let \(\widehat N=\operatorname{Hom}(N,U(1))\) denote the character group of \(N\). For \(\nu\in\widehat N\), define the joint spectral projector
\begin{equation}
 P_\nu^{(N)}=d^{-k}\sum_{u\in N}\nu(u)\,w(u).
 \label{eq:stabilizer-sector-projector}
\end{equation}
The projectors \(\{P_\nu^{(N)}\}_{\nu\in\widehat N}\) are mutually orthogonal, sum to the identity, and have rank \(d^{n-k}\). They are rank one only when \(N\) is maximal isotropic, namely \(k=n\).

The stabilizer pinching associated with \(N\) is
\begin{align}\label{eq:correct-pinching}
 &\Delta_N(\rho)=\sum_{\nu\in\widehat N}P_\nu^{(N)}\rho P_\nu^{(N)} \nonumber \\
 =&\frac{1}{d^k}\sum_{u\in N}w(u)\rho w(u)^\dagger=d^{-n}\sum_{z\in N^\perp}\Xi_\rho(z)w(z),
\end{align}
where \(N^\perp=\{z:\langle z,u\rangle_s=0\ \forall u\in N\}\). Equation \eqref{eq:correct-pinching} is important: a partial pinching retains the Weyl modes in \(N^\perp\), not only those in \(N\). For maximal isotropic \(N\), one has \(N^\perp=N\).

\subsection{Exact stabilizer skeleton and mean state}

For a state \(\rho\), define its exact stabilizer skeleton
\begin{equation*}
 S_\rho:=\{z\in\Vn:|\Xi_\rho(z)|=1\}.
\end{equation*}
The mean state \(M(\rho)\) is specified by
\begin{equation*}
 \Xi_{M(\rho)}(z)=
 \begin{cases}
 \Xi_\rho(z),& |\Xi_\rho(z)|=1,\\
 0,&|\Xi_\rho(z)|<1.
 \end{cases}
\end{equation*}
As reviewed below, \(S_\rho\) is isotropic and \(M(\rho)\) is the maximally mixed state on the stabilizer sector fixed by the exact constraints.

\subsection{Quantum convolutional channels}

Let
\begin{equation*}
 G=\begin{pmatrix}g_{00}&g_{01}\\ g_{10}&g_{11}\end{pmatrix}
\end{equation*}
be invertible over \(\Zd\), and set \(N_G=(\det G)^{-1}\). The convolutional unitary is
\begin{equation}
 U_G=\sum_{i,j}
 \ket{N_Gg_{11}i-N_Gg_{10}j,-N_Gg_{01}i+N_Gg_{00}j}\bra{i,j}.
 \label{eq:UG}
\end{equation}
For an environment state \(\sigma\), define
\begin{equation*}
 \Lambda_\sigma(\rho)=\Tr_B[U_G(\rho\otimes\sigma)U_G^\dagger].
\end{equation*}
The convolution--multiplication duality reads \cite{Bu2023pnas}
\begin{equation}
 \Xi_{\Lambda_\sigma(\rho)}(p,q)
 =\Xi_\rho(N_Gg_{11}p,g_{00}q)
  \Xi_\sigma(-N_Gg_{10}p,g_{01}q).
 \label{eq:duality}
\end{equation}
We call \(G\) even-parity positive if \(g_{00},g_{11}\neq0\), odd-parity positive if \(g_{01},g_{10}\neq0\), and positive if all four entries are nonzero.

The one-shot Holevo information and regularized Holevo capacity are
\begin{align*}
 &\chi(\Lambda)=\max_{\{p_i,\rho_i\}}
 S\!\left(\Lambda\!\left(\sum_i p_i\rho_i\right)\right)
 -\sum_i p_iS(\Lambda(\rho_i)),\\
 &C_H(\Lambda)=\lim_{m\to\infty}\frac{1}{m}\chi(\Lambda^{\otimes m}),
\end{align*}
where \(S(\rho)\) denotes the von Neumann entropy of $\rho$. Similarly, \(H( p)\) denotes the Shannon entropy of vector $p$. 

Denote the complementary channel of $\Lambda$ by $\Lambda^c$, the coherent information is
\begin{equation*}
 I_c(\rho,\Lambda)=S(\Lambda(\rho))-S(\Lambda^c(\rho)).
\end{equation*}
The quantum capacity satisfies $Q(\Lambda)\ge\max_\rho I_c(\rho,\Lambda)$ \cite{Lloyd1997pra,Devetak2005tit,Devetak2005cmp}.

\section{Covariance and capacity preliminaries}

\begin{theorem}[Weyl covariance]
For every invertible convolution matrix \(G\),
\begin{align}
 &\Lambda_\sigma[w(a,b)\rho w(a,b)^\dagger]\nonumber\\
 =&w(g_{00}a,N_Gg_{11}b)\Lambda_\sigma(\rho)
  w(g_{00}a,N_Gg_{11}b)^\dagger.
 \label{eq:covariance}
\end{align}
\end{theorem}

\begin{proof}
Conjugation of the Weyl expansion \eqref{eq:weyl-expansion} by \(w(a,b)\) multiplies the coefficient of \(w(p,q)\) by \(\omega_d^{a\cdot q-b\cdot p}\). Equation \eqref{eq:duality} transfers this phase to \(\omega_d^{g_{00}a\cdot q-N_Gg_{11}b\cdot p}\), which is precisely the phase generated by conjugation of the output by \(w(g_{00}a,N_Gg_{11}b)\).
\end{proof}

\begin{corollary}[Minimal-output-entropy form]
\label{cor:moe}
If \(G\) is even-parity positive, then \(\Lambda_\sigma\) is unital and
\begin{equation}
 \chi(\Lambda_\sigma)=n\log d-
 \min_{\ket\psi}S\!\left(\Lambda_\sigma(\proj\psi)\right).
 \label{eq:moe}
\end{equation}
\end{corollary}

\begin{proof}
The map \((p,q)\mapsto(N_Gg_{11}p,g_{00}q)\) is invertible, so Eq.~\eqref{eq:duality} sends the characteristic function of \(\id/d^n\) to itself. Thus the channel is unital. The upper bound in Eq.~\eqref{eq:moe} follows from the maximum output entropy and minimum output entropy. Conversely, take a pure minimizer and average its Weyl orbit. Covariance keeps all output entropies equal, while the uniform Weyl orbit averages to \(\id/d^n\), attaining the upper bound.
\end{proof}

The entanglement-breaking statement used below was proved in Ref.~\cite{Xiong2026Private}. We record only its classical-capacity consequence.

\begin{theorem}[Strong-additivity consequence for stabilizer environments]
\label{thm:stabilizer-additivity}
Let \(G\) be positive. If \(\sigma\) is a convex combination of pure stabilizer states, then \(\Lambda_\sigma\) is entanglement-breaking. Consequently, for every channel \(\Phi\),
\begin{equation}
 \chi(\Lambda_\sigma\otimes\Phi)=\chi(\Lambda_\sigma)+\chi(\Phi), C_H(\Lambda_\sigma)=\chi(\Lambda_\sigma).
 \label{eq:strong-additivity}
\end{equation}
\end{theorem}

\begin{proof}
For a pure stabilizer environment, Ref.~\cite{Xiong2026Private} gives an explicit measure-and-prepare representation. Convexity preserves the entanglement-breaking property. Eq.\eqref{eq:strong-additivity} then follows from strong additivity of the Holevo information for entanglement-breaking channels \cite{Shor2002jmp}.
\end{proof}

\section{Mean-state entropy hierarchy}

This section contains the main structural result of the paper: the mean-state entropy hierarchy. The maximal-entropy endpoint is known \cite{Bu2023pnas,Bu2023discrete}; the new content is the exact hierarchy of compatible stabilizer pinchings and its relative-entropy decomposition.

\subsection{Exact constraints and normal form}

\begin{lemma}[Unit-modulus expectation]
Let \(\rho\) be a state and \(U\) a unitary. If \(|\Tr(\rho U)|=1\), then the support of \(\rho\) lies in one eigenspace of \(U\). More precisely, for \(\lambda=\Tr(\rho U)\),
\begin{equation}
 U\rho=\lambda\rho,\qquad \rho U=\lambda\rho.
\end{equation}
\end{lemma}

\begin{proof}
The expectation value is a convex combination of unit-modulus eigenvalues of \(U\). Its modulus can be one only if every eigenvalue carrying nonzero weight is the same.
\end{proof}

\begin{lemma}[Structure of the exact skeleton]\label{lem:exact-skeleton}
The set \(S_\sigma\) is a linear isotropic subspace of \(\Vn\), and \(\vartheta_\sigma:=\Xi_\sigma|_{S_\sigma}\) is a character of its additive group.
\end{lemma}

\begin{proof}
For \(z\in S_\sigma\), the preceding lemma gives
\begin{equation}
 w(-z)\sigma=\Xi_\sigma(z)\sigma.
 \label{eq:eigenconstraint}
\end{equation}
If \(z,z'\in S_\sigma\), a nonzero vector in \(\supp\sigma\) is a simultaneous eigenvector of \(w(-z)\) and \(w(-z')\). The Weyl commutation relation then forces \(\langle z,z'\rangle_s=0\). On an isotropic subspace the Weyl cocycle is trivial, so Eq.~\eqref{eq:eigenconstraint} implies closure under addition and scalar multiplication and proves multiplicativity of \(\vartheta_\sigma\).
\end{proof}

\begin{proposition}[Clifford normal form and maximal-entropy endpoint \cite{Bu2023pnas,Bu2023discrete}]
Let \(r=\dim S_\sigma\). There exist a Clifford unitary \(C\) and an \((n-r)\)-qudit state \(\tau_E\) such that
\begin{align}
& C\sigma C^\dagger=\proj{0}^{\otimes r}\otimes\tau_E,\label{eq:sigma-normal}\\
& CM(\sigma)C^\dagger=\proj{0}^{\otimes r}\otimes\frac{\id_E}{d^{n-r}},
 \label{eq:mean-normal}
\end{align}
where \(E\) denotes the residual \((n-r)\)-qudit subsystem. Hence \(S(M(\sigma))=(n-r)\log d\). Moreover, \(M(\sigma)\) is the unique entropy maximizer among all states satisfying the exact constraints \eqref{eq:eigenconstraint}.
\end{proposition}

\begin{proof}
By Lemma~\ref{lem:exact-skeleton}, \(S_\sigma\) is a linear isotropic subspace. A symplectic Gram--Schmidt construction maps \(S_\sigma\) to the span of the first \(r\) computational \(Z\)-generators. A Weyl displacement removes the residual character. The exact constraints then force the first \(r\) qudits into \(\ket0^{\otimes r}\), proving Eq.~\eqref{eq:sigma-normal}. Clifford covariance of the characteristic function gives Eq.~\eqref{eq:mean-normal}. Every state obeying the same constraints is supported on a subspace of dimension \(d^{n-r}\), whose unique entropy maximizer is its normalized projector.
\end{proof}

\subsection{Pinching Pythagoras and compatible coarse-grainings}

\begin{lemma}[Pinching Pythagoras]
\label{lem:pinching-pythagoras}
Let \(\Delta\) be an orthogonal pinching and set \(\rho_\Delta=\Delta(\rho)\). If \(\omega\) is fixed by \(\Delta\) and \(\supp\rho\subseteq\supp\omega\), then
\begin{equation}
 D(\rho\|\omega)=D(\rho\|\rho_\Delta)+D(\rho_\Delta\|\omega),
 \label{eq:pinching-pythagoras}
\end{equation}
where
\begin{equation}
 D(\rho\|\rho_\Delta)=S(\rho_\Delta)-S(\rho).
 \label{eq:pinching-entropy}
\end{equation}
\end{lemma}

\begin{proof}
The operator \(\log\rho_\Delta\) is fixed by \(\Delta\), and \(\Delta\) is self-adjoint for the Hilbert--Schmidt inner product. Therefore
\(\Tr(\rho\log\rho_\Delta)=\Tr(\rho_\Delta\log\rho_\Delta)\), proving Eq.~\eqref{eq:pinching-entropy}. The same fixed-point identity for \(\log\omega\) gives Eq.~\eqref{eq:pinching-pythagoras}.
\end{proof}

We call an isotropic subspace \(N\) \emph{compatible} with \(\sigma\) when \(S_\sigma\subseteq N\).

\begin{theorem}[Compatible-dephasing decomposition]
\label{thm:compatible-dephasing}
Let \(N\supseteq S_\sigma\) be isotropic with \(\dim N=k\). A Clifford unitary \(C_N\) can be chosen such that Eqs.~\eqref{eq:sigma-normal} and \eqref{eq:mean-normal} hold and
\begin{equation}
 C_N\Delta_N(\sigma)C_N^\dagger
 =\proj{0}^{\otimes r}\otimes
 \left(\Delta_Z^{\otimes(k-r)}\otimes\mathrm{\id}^{\otimes(n-k)}\right)(\tau_E),
 \label{eq:compatible-normal}
\end{equation}
where \(\Delta_Z\) is the single-qudit computational-basis dephasing channel and \(\mathrm{\id}\) denotes the identity channel on one qudit. Consequently,
\begin{equation}
 S(\sigma)\le S(\Delta_N(\sigma))\le S(M(\sigma)),
 \label{eq:compatible-order}
\end{equation}
with the exact identities
\begin{align}
 &D(\sigma\|\Delta_N(\sigma))
 =S(\Delta_N(\sigma))-S(\sigma),\notag\\
 &D(\Delta_N(\sigma)\|M(\sigma))
 =S(M(\sigma))-S(\Delta_N(\sigma)),\notag\\
 &D(\sigma\|M(\sigma))
 =D(\sigma\|\Delta_N(\sigma))
  +D(\Delta_N(\sigma)\|M(\sigma)).
 \label{eq:two-step-pythagoras}
\end{align}
\end{theorem}

\begin{proof}
By Lemma~\ref{lem:exact-skeleton}, \(S_\sigma\) is isotropic.  
Extend a basis of \(S_\sigma\) first to a basis of \(N\) and then to a symplectic basis of \(\Vn\). The induced Clifford maps the pair \(S_\sigma\subseteq N\) to
\begin{align*}
 S_0&=\mathrm{span}\{(e_1,0),\ldots,(e_r,0)\},\\
 N_0&=\mathrm{span}\{(e_1,0),\ldots,(e_k,0)\}.
\end{align*}
After a Weyl displacement trivializes the exact character, \(\sigma\) takes the form \eqref{eq:sigma-normal}. 
The twirling form in Eq.~\eqref{eq:correct-pinching} then shows what the pinching does in this normal form. The first \(r\) generators of \(N_0\) act trivially on the factor \(|0\rangle\langle0|^{\otimes r}\), because these qudits are already fixed by the exact stabilizer constraints. The remaining \(k-r\) generators of \(N_0/S_0\) are the computational \(Z\)-generators on the first \(k-r\) qudits of the residual system \(E\). Hence the pinching dephases exactly these \(k-r\) residual qudits and leaves the last \(n-k\) residual qudits untouched. This proves Eq.~\eqref{eq:compatible-normal}.
Random-unitary pinching cannot decrease entropy, and the output remains supported on the same \(d^{n-r}\)-dimensional exact stabilizer sector. This proves Eq.~\eqref{eq:compatible-order}. Finally, apply Lemma~\ref{lem:pinching-pythagoras} with \(\omega=M(\sigma)\), which is fixed by every compatible pinching and is uniform on the common support sector.
\end{proof}

\subsection{The hierarchy}

\begin{theorem}[Mean-state entropy hierarchy]
\label{thm:hierarchy}
Let
\begin{equation*}
 S_\sigma=N_r\subset N_{r+1}\subset\cdots\subset N_n=\cM
\end{equation*}
be a compatible isotropic flag with \(\dim N_k=k\), and set
\(\rho_k=\Delta_{N_k}(\sigma)\). Then
\begin{equation*}
 S(\sigma)=S(\rho_r)\le S(\rho_{r+1})\le\cdots\le S(\rho_n)\le S(M(\sigma)).
\end{equation*}
For every \(r\le k<n\),
\begin{align}
 &D(\rho_k\|\rho_{k+1})
 =S(\rho_{k+1})-S(\rho_k),\notag\\
& D(\rho_k\|M(\sigma))
 =D(\rho_k\|\rho_{k+1})+D(\rho_{k+1}\|M(\sigma)).
 \label{eq:nested-pythagoras}
\end{align}
Thus
\begin{equation}
 D(\sigma\|M(\sigma))
 =\sum_{k=r}^{n-1}D(\rho_k\|\rho_{k+1})
  +D(\rho_n\|M(\sigma)).
 \label{eq:full-pythagoras}
\end{equation}
Each one-generator step obeys
\begin{equation}
 0\le S(\rho_{k+1})-S(\rho_k)\le\log d.
 \label{eq:one-step-bound}
\end{equation}
\end{theorem}

\begin{proof}
Equation \eqref{eq:eigenconstraint} implies \(\Delta_{S_\sigma}(\sigma)=\sigma\). Nested twirls satisfy the absorption identity
\begin{equation*}
 \Delta_{N_{k+1}}\circ\Delta_{N_k}=\Delta_{N_{k+1}},
\end{equation*}
so \(\rho_{k+1}=\Delta_{N_{k+1}}(\rho_k)\). 


Since \(M(\sigma)\) is the normalized projector onto the exact stabilizer
sector determined by \(S_\sigma\), it is fixed by every compatible pinching \(\Delta_{N_k}\). Therefore Lemma~\ref{lem:pinching-pythagoras} applies with
\(\rho=\rho_k\), \(\rho_\Delta=\rho_{k+1}\), and \(\omega=M(\sigma)\), yielding Eq.~\eqref{eq:nested-pythagoras}. 
Iterating the second identity in Eq.~\eqref{eq:nested-pythagoras} from \(k=r\) to \(k=n-1\), and using \(\rho_r=\sigma\), gives Eq.~\eqref{eq:full-pythagoras}.



For Eq.~\eqref{eq:one-step-bound}, write \(N_{k+1}=N_k+\mathrm{span}\{v_{k+1}\}\). On states fixed by \(\Delta_{N_k}\),
\begin{equation*}
 \rho_{k+1}=\frac1d\sum_{t=0}^{d-1}w(tv_{k+1})\rho_kw(tv_{k+1})^\dagger.
\end{equation*}
The entropy of a mixture of \(d\) unitarily equivalent states is at most their common entropy plus \(\log d\).
\end{proof}

The entropy hierarchy and its relative entropy decomposition are illustrated in FIG.\ref{fig:hierarchy}.

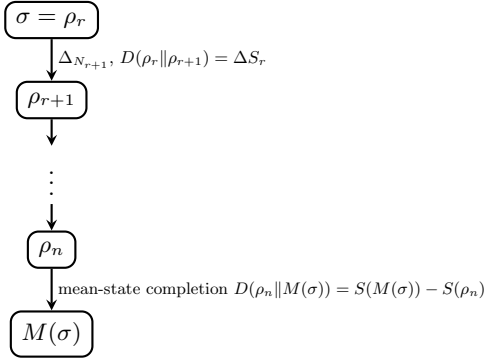
\begin{figure}[t]
\centering
\begin{tikzpicture}[>=stealth,thick,node distance=0.55cm]
\node[draw,rounded corners,inner sep=4pt] (s) {$\sigma=\rho_r$};
\node[draw,rounded corners,inner sep=4pt,below=of s] (r1) {$\rho_{r+1}$};
\node[below=0.35cm of r1] (dots) {$\vdots$};
\node[draw,rounded corners,inner sep=4pt,below=0.35cm of dots] (rn) {$\rho_n$};
\node[draw,rounded corners,inner sep=4pt,below=of rn] (m) {$M(\sigma)$};
\draw[->] (s)--node[right,scale=.72] {$\Delta_{N_{r+1}}$, $D(\rho_r\|\rho_{r+1})=\Delta S_r$} (r1);
\draw[->] (r1)--(dots);
\draw[->] (dots)--(rn);
\draw[->] (rn)--node[right,align=left,scale=.70] {mean-state completion\ $D(\rho_n\|M(\sigma))=S(M(\sigma))-S(\rho_n)$} (m);
\end{tikzpicture}
\caption{Here \(\Delta S_k:=S(\rho_{k+1})-S(\rho_k)\), and each step satisfies
\(D(\rho_k\|\rho_{k+1})=\Delta S_k\). The terminal gap satisfies
\(D(\rho_n\|M(\sigma))=S(M(\sigma))-S(\rho_n)\). Every compatible isotropic flag gives a finite-step entropy hierarchy below the mean-state ceiling. Entropy is non-decreasing, each pinching step adds at most $\log d$, and the terminal term measures residual classical non-uniformity inside the exact stabilizer sector.}
\label{fig:hierarchy}
\end{figure}

\begin{corollary}[Optimized entropy profile]
For \(k=r,\ldots,n\), define
\begin{equation}
\mathfrak{s}_k(\sigma)=\min_{\substack{N\supseteq S_\sigma,\ N\ \mathrm{isotropic}\\ \dim N=k}}
 S(\Delta_N(\sigma)).
\label{eq:hk}
\end{equation}
Then
\begin{equation*}
 S(\sigma)=\mathfrak{s}_r(\sigma)\le \mathfrak{s}_{r+1}(\sigma)\le\cdots\le \mathfrak{s}_n(\sigma)\le S(M(\sigma)),
\end{equation*}
with \(0\le \mathfrak{s}_{k+1}(\sigma)-\mathfrak{s}_k(\sigma)\le\log d\).
\end{corollary}

\begin{proof}
Every compatible \((k+1)\)-dimensional isotropic subspace contains a compatible \(k\)-dimensional one, so Theorem~\ref{thm:hierarchy} gives monotonicity after minimization. Conversely, let \(N_k\) be an optimizer at level \(k<n\). Since \(N_k\)
is a non-maximal isotropic subspace of the finite symplectic space \(V^n\),
it can be extended by one isotropic generator: there exists
\(v\in N_k^\perp\setminus N_k\) such that
\[
N_{k+1}:=N_k+\operatorname{span}\{v\}
\]
is isotropic. Moreover, \(N_{k+1}\) still contains \(S_\sigma\). Applying
Eq.~\eqref{eq:one-step-bound} to this one-generator extension gives
\[
s_{k+1}(\sigma)
\le S(\Delta_{N_{k+1}}(\sigma))
\le S(\Delta_{N_k}(\sigma))+\log d
=
s_k(\sigma)+\log d .
\]
\end{proof}

\section{Stabilizer spectral transfer and communication}


The hierarchy itself is a state-space statement. The spectral-transfer argument yields a strengthening in the communication setting: compatibility with the exact stabilizer skeleton is not required for the communication lower bound itself. The construction below applies to
arbitrary complete stabilizer dephasings.



\begin{theorem}[Spectral-transfer Holevo lower bound]\label{thm:holevo-bound}
Let \(G\) be positive and let \(\sigma\) be an \(n\)-qudit state. Then
\begin{align}\label{eq:holevo-sandwich}
n\log d-
 \min_{\cM}
 S(\Delta_{\cM}(\sigma))
 \le
 \chi(\Lambda_\sigma)
 \le
 n\log d-S(\sigma),   
\end{align}
where the minimum is taken over all maximal isotropic subspaces. Moreover, the mean-state hierarchy gives the canonical compatible bound
\begin{equation}
 \chi(\Lambda_\sigma)
 \ge n\log d-\sprof_n(\sigma)
 \ge n\log d-S(M(\sigma)).
 \label{eq:mean-bound-refined}
\end{equation}
\end{theorem}

\begin{proof}
The entropy inequality for quantum convolution \cite{Bu2023discrete} gives
\(S(\Lambda_\sigma(\proj\psi))\ge S(\sigma)\) for every pure input. Together with Eq.~\eqref{eq:moe}, this proves the upper bound.

Fix a maximal isotropic subspace \(\mathcal{M}\). No compatibility assumption $\mathcal{M}\supset S_{\sigma}$ is required here because the construction only relies on complete stabilizer measurements. Define
\begin{align*}
 A_G(p,q)=(N_Gg_{11}p,g_{00}q),B_G(p,q)=(-N_Gg_{10}p,g_{01}q),
\end{align*}
and set
\[
K_\cM:=B_G^{-1}(\cM),\qquad L_\cM:=A_G(K_\cM).
\]
By Lemma~\ref{lem:AGBG-similitudes}, positivity of \(G\) makes both \(A_G\) and \(B_G\) invertible symplectic similitudes. Hence, preimages and images of maximal isotropic subspaces under these maps are again maximal isotropic, so \(K_\cM\) and \(L_\cM\) are maximal isotropic.


On one hand, for \(\mu\in\widehat \cM\), define the joint spectral projector $P_\mu^{(\cM)}=d^{-n}\sum_{u\in \cM}\mu(u)\,w(u),$
\begin{align*}
 \Delta_\cM(\sigma)=\sum_{\mu\in\widehat\cM}p_\mu P_\mu^{(\cM)}=d^{-n}\sum_{\mu\in\widehat\cM}p_\mu\sum_{m\in\cM}\mu(m)w(m).
\end{align*}
Comparing with Eq.(\ref{eq:correct-pinching}) gives $\Xi_{\sigma}(m)=\sum_{\mu\in\widehat{\cM}}\mu(m)p_{\mu}$ for each $m\in\cM$.
On the other hand, choose the pure stabilizer input \(Q_\eta^{(L_\mathcal{M})}\) associated with a character \(\eta\in\widehat{L_\mathcal{M}}\).
Equation~\eqref{eq:duality} restricts the output Weyl expansion to \(K_\mathcal{M}\).
Substitution of the Fourier expansion of \(p_\mu\) yields
\begin{align*}\label{eq:spectral-transfer}
    \Lambda_\sigma(Q_\eta^{(L_\cM)})=&d^{-n}\sum_{x\in K_\cM}\eta(A_Gx)\Xi_{\sigma}(B_Gx)w(x)\\
    =&\sum_{\mu\in\widehat\cM}p_{\mu}d^{-n}\sum_{x\in K_\cM}\eta(A_Gx)\mu(B_Gx)w(x).
\end{align*}
Denote \(R_{\eta,\mu}^{(K_\cM)}:=d^{-n}\sum_{x\in K_\cM}\eta(A_Gx)\mu(B_Gx)w(x)\), then 
\(\{R_{\eta,\mu}^{(K_\cM)}\}_\mu\) is a complete orthogonal family of rank-one stabilizer projectors and they satisfy
$\sum_\mu R_{\eta,\mu}^{(K_\cM)}=I.$ Distinct \(\mu\) give distinct characters because \(B_G:K_\cM\to\cM\) is an isomorphism. Hence 
\begin{equation*}
 S(\Lambda_\sigma(Q_\eta^{(L_\cM)}))=S(\Delta_\cM(\sigma)).
\end{equation*}
The minimum output entropy is therefore no larger than the right-hand side, and optimization over \(\cM\), followed by Eq.~\eqref{eq:moe}, proves the lower bound in Eq.~\eqref{eq:holevo-sandwich}. 

The compatible bound follows because the minimization defining \(\sprof_n(\sigma)\) is restricted to maximal isotropic subspaces containing the exact stabilizer skeleton, where \(\sprof_n(\sigma)\) is the endpoint of the compatible entropy profile defined in Eq.(\ref{eq:hk}).
\end{proof}


The bound is operationally measurable: for each maximal isotropic \(\cM\), \(S(\Delta_\cM(\sigma))\) is exactly the Shannon entropy of the outcome distribution of the associated stabilizer measurement. The theorem therefore converts a finite set of environment measurements into certified classical communication rates.

\section{Exact, robust, and single-letter formulae}

\begin{theorem}[Stabilizer-diagonal environments]
\label{thm:stabilizer-diagonal}
Let \(G\) be positive. If \(\sigma\) is diagonal in the stabilizer basis for a maximal isotropic subspace \(\mathcal M_0\), namely
\[
\sigma=\sum_{\mu\in\widehat{\cM_0}}\lambda_\mu P_\mu^{(\cM_0)},
\]
then
\begin{align*}
 C_H(\Lambda_\sigma)&=\chi(\Lambda_\sigma)=n\log d-S(\sigma)\notag\\
 &=n\log d-\min_{\substack{\cM\ {\rm maximal}\\ {\rm  isotropic}}}S(\Delta_\cM(\sigma)).
\end{align*}
\end{theorem}

\begin{proof}
Pinching never decreases entropy, while \(\Delta_{\cM_0}(\sigma)=\sigma\). Thus the dephasing minimum equals \(S(\sigma)\), and Theorem~\ref{thm:holevo-bound} gives the one-shot equality. Stabilizer-diagonal states are convex mixtures of pure stabilizer states, so Theorem~\ref{thm:stabilizer-additivity} gives additivity.
\end{proof}

The minimum is attained at the stabilizer basis diagonalizing \(\sigma\). Exact additivity need not be stable under perturbation, but the one-shot formula has a robust neighborhood.

\begin{theorem}[Robust one-shot neighborhood]
\label{thm:robust}
Let G be positive and let \(\tau\) be stabilizer diagonal, and suppose
\begin{equation*}
 \frac12\|\sigma-\tau\|_1\le\varepsilon\le1-\frac1{d^n}.
\end{equation*}
Then
\begin{equation}
 \left|\chi(\Lambda_\sigma)-\bigl[n\log d-S(\tau)\bigr]\right|
 \le \varepsilon\log({d^n}-1)+\hbin(\varepsilon).
 \label{eq:robust-holevo}
\end{equation}
\end{theorem}

\begin{proof}
For every pure input \(\proj\psi\), unitary invariance and contractivity of trace distance under partial trace give
\begin{equation*}
 \frac12\|\Lambda_\sigma(\proj\psi)-\Lambda_\tau(\proj\psi)\|_1\le\varepsilon.
\end{equation*}
The sharp entropy continuity bound \cite{Audenaert2007jpa} therefore implies that the two output entropies differ by at most the right hand of Eq.(\ref{eq:robust-holevo}), uniformly in \(\psi\). Their minimum output entropies differ by the same amount. Equation \eqref{eq:moe} and Theorem~\ref{thm:stabilizer-diagonal} complete the proof.
\end{proof}

The regularized capacity of a nearby nonstabilizer channel may still be nonadditive; Eq.~\eqref{eq:robust-holevo} is deliberately a one-shot statement.

\begin{problem}[Single-letter stabilizer-dephasing formula]
For which environment states is the minimum output entropy attained by a pure stabilizer input, equivalently,
\begin{equation}
 \chi(\Lambda_\sigma)
 =n\log d-\min_{\cM}S(\Delta_\cM(\sigma))?
 \label{eq:single-letter-problem}
\end{equation}
where the minimum is taken over all maximal isotropic subspaces. Even when Eq.~\eqref{eq:single-letter-problem} holds, additivity is a separate question.
\end{problem}

\subsection{An exact qutrit nonstabilizer family}

Let \(d=3\), \(n=1\), and
\begin{equation*}
 G_H=\begin{pmatrix}1&1\\1&-1\end{pmatrix}.
\end{equation*}
For
\begin{align}
 \ket{\psi_p}=\sqrt p\ket0+\sqrt{1-p}\ket1,\sigma_p=\proj{\psi_p},0\le p\le1.
 \label{eq:qutrit-family}
\end{align}
write \(q=1-p\).

\begin{theorem}[Exact qutrit single-letter formula]
For the family \eqref{eq:qutrit-family},
\begin{equation}
 \chi(\Lambda_{\sigma_p})
 =\log3-\hbin(p)
 =\log3-\min_{\substack{\cM\ {\rm maximal}\\ {\rm  isotropic}}}S(\Delta_\cM(\sigma_p)),
 \label{eq:qutrit-holevo}
\end{equation}
where \(\hbin(p)\) is the binary Shannon entropy.
\end{theorem}


\begin{proof}
Computational-basis dephasing gives probabilities \((p,q,0)\), hence entropy \(\hbin(p)\). The other three qutrit stabilizer bases are mutually unbiased with respect to the computational basis, as is expected for the complete set of qutrit stabilizer measurements \cite{Durt2010IJTP}. Their probabilities have the form
\begin{equation*}
 r_k(\theta)=\frac{1+2\sqrt{pq}\cos(\theta+2\pi k/3)}{3},k=0,1,2.
\end{equation*}
Since Shannon entropy dominates collision entropy,
\begin{align*}
 H(r_0,r_1,r_2)
 &\ge-\log\sum_{k=0}^2r_k^2=\log3-\log(1+2pq)\ge\hbin(p).
\end{align*}
The last inequality is proved in Appendix~\ref{app:qutrit-entropy}. Hence the dephasing minimum is \(\hbin(p)\).

For an arbitrary pure input \(\ket\phi=a\ket0+b\ket1+c\ket2\), let \(x=|a|^2\), \(y=|b|^2\), and \(z=|c|^2\). Direct evaluation of \(U_{G_H}\) gives
\begin{align*}
 \Tr\!\left[\Lambda_{\sigma_p}(\proj\phi)^2\right]
 &=p^2+q^2-2(p-q)^2(xy+yz+zx)\notag\\
 &\le p^2+q^2.
\end{align*}
Applying the qutrit spectral lemma in Appendix~\ref{app:qutrit-entropy} with \(r=\max\{p,q\}\), any qutrit state with purity at most \(p^2+q^2=r^2+(1-r)^2\) has entropy at least \(h_2(r)=h_2(p)\).
Computational-basis inputs attain the rank-two output spectrum \((p,q,0)\). Therefore the minimum output entropy is \(\hbin(p)\), and Eq.~\eqref{eq:moe} proves Eq.~\eqref{eq:qutrit-holevo}.
\end{proof}

For every \(0<p<1\), the exact skeleton of \(\sigma_p\) is trivial, so \(M(\sigma_p)=\id/3\) and the earlier mean-state bound is zero. Equation \eqref{eq:qutrit-holevo} is therefore a strict refinement throughout the open interval.

\section{Coexistence of classical and quantum communication}

The qutrit family also yields a complete analytic coherent-information sign pattern for two fixed input states. Define
\begin{align*}
& \rho_0=\frac34\proj0+\frac14\proj1
 +\frac14\ketbra{0}{1}+\frac14\ketbra{1}{0},\\ 
& R\ket j=\ket{1-j}, \rho_1=R\rho_0R^\dagger.
\end{align*}
The unitary \(R\) is a qutrit Clifford permutation.

\begin{theorem}[Coherent-information sign pattern for a fixed input pair]
Let
\begin{align}
 \bm\lambda(p)&=\bigl(\lambda_0(p),\lambda_+(p),\lambda_-(p)\bigr),
 \label{eq:lambda-p}
\end{align}
with $\lambda_0(p)=\frac{3-2p}{4}$ and $\lambda_\pm(p)=\frac{1+2p\pm\sqrt{12p^2-4p+1}}{8}$. Then
\begin{align}
 I_c(\rho_0,\Lambda_{\sigma_p})
 &=H(\bm\lambda(p))-H(\bm\lambda(1-p)),\label{eq:ic0}\\
 I_c(\rho_1,\Lambda_{\sigma_p})
 &=-I_c(\rho_0,\Lambda_{\sigma_p}).
 \label{eq:ic1}
\end{align}
Moreover,
\begin{equation}
 I_c(\rho_0,\Lambda_{\sigma_p})
 \begin{cases}>0,&0<p<\frac12,\\=0,&p\in\{0,\frac12,1\},\\<0,&\frac12<p<1.
 \end{cases}
 \label{eq:ic-sign}
\end{equation}
Consequently,
\[
I_c^{\rm lb}(p):=
\max\{I_c(\rho_0,\Lambda_{\sigma_p}),
I_c(\rho_1,\Lambda_{\sigma_p})\}
=
|I_c(\rho_0,\Lambda_{\sigma_p})|
\]
is positive for every \(p\in(0,1)\setminus\{1/2\}\). At \(p=1/2\),
\begin{equation}
 \Lambda_{\sigma_{1/2}}^c(\rho)=X\Lambda_{\sigma_{1/2}}(\rho)X^\dagger
 \label{eq:self-complementary}
\end{equation}
for every input \(\rho\). Hence the channel is unitarily self-complementary and
\(Q(\Lambda_{\sigma_{1/2}})=0\). At \(p=0,1\), the environment is stabilizer and the channel is entanglement breaking, so its quantum capacity also vanishes.
\end{theorem}

\begin{proof}
Direct calculation gives
\begin{align*}
 \Lambda_{\sigma_p}(\rho_0)
 &=\begin{pmatrix}
 3p/4&\sqrt{pq}/4&0\\
 \sqrt{pq}/4&q/4&0\\
 0&0&(3-2p)/4
 \end{pmatrix},\\
 \Lambda_{\sigma_p}^c(\rho_0)
 &=\begin{pmatrix}
 (1+2p)/4&0&0\\
 0&3q/4&\sqrt{pq}/4\\
 0&\sqrt{pq}/4&p/4
 \end{pmatrix}.
\end{align*}
Their spectra are \(\bm\lambda(p)\) and \(\bm\lambda(1-p)\), respectively, proving Eq.~\eqref{eq:ic0}. For \(\rho_1\), the two spectra are exchanged, proving Eq.~\eqref{eq:ic1}.

The spectra \(\bm\lambda(p)\) and \(\bm\lambda(1-p)\) have the same purity,
\begin{equation*}
 P_2(p)=\frac{6p^2-6p+5}{8}=P_2(1-p),
\end{equation*}
and determinants
\begin{equation*}
 \det\bm\lambda(p)=\frac{p(1-p)(3-2p)}{32},
\end{equation*}
so
\begin{equation*}
 \det\bm\lambda(p)-\det\bm\lambda(1-p)
 =\frac{p(1-p)(1-2p)}{16}.
\end{equation*}
Appendix~\ref{app:qutrit-entropy} proves that, for qutrit spectra with fixed trace and purity, the entropy is strictly increasing in the determinant. Equation \eqref{eq:ic-sign} follows.

At \(p=1/2\), Eq.~\eqref{eq:self-complementary} shows that the channel and its complement are unitarily equivalent. Hence the inverse unitary also degrades the complementary output back to the channel output, so the channel is antidegradable and has zero quantum capacity.
The endpoint statement follows from Theorem~\ref{thm:stabilizer-additivity}.
\end{proof}

\begin{figure}[t]
\centering
\includegraphics[width=\columnwidth]{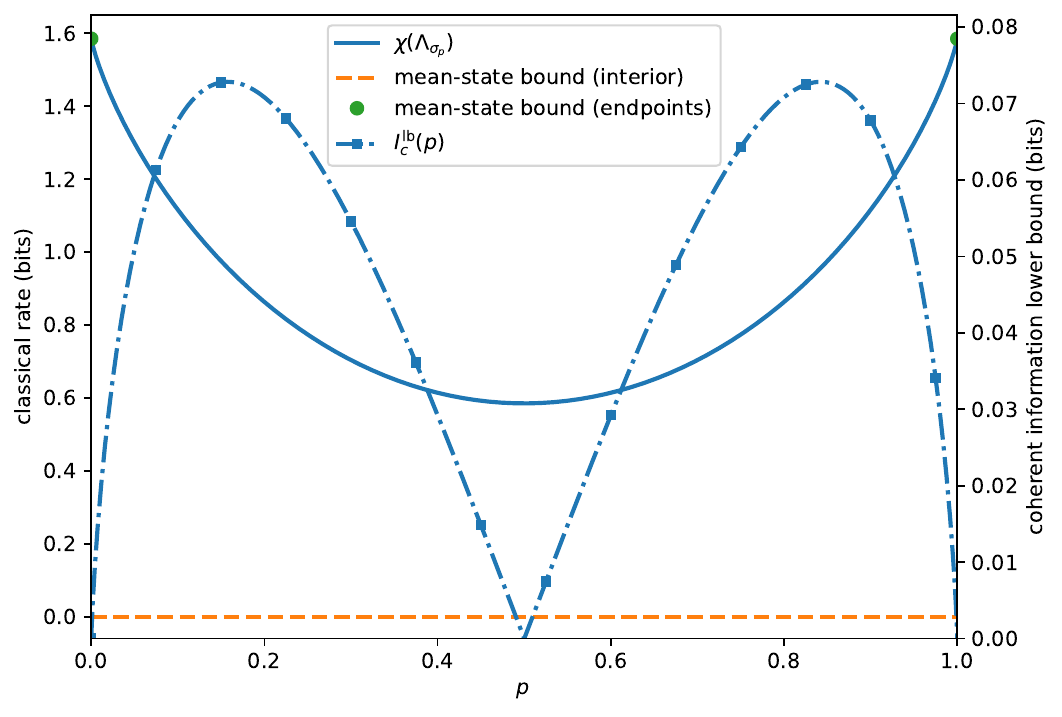}
\caption{The exact one-shot Holevo information \(\chi(\Lambda_{\sigma_p})=\log3-\hbin(p)\) and the explicit coherent-information lower bound \(I_c^{\rm lb}(p)=|I_c(\rho_0,\Lambda_{\sigma_p})|\). The right axis is used for the coherent-information curve. For \(0<p<1\), the old mean-state lower bound is zero because \(M(\sigma_p)=\id/3\); the isolated stabilizer endpoints have mean-state bound \(\log3\). Classical communication remains positive for all \(p\), while the displayed quantum lower bound vanishes at \(p=0,1/2,1\).}
\label{fig:qutrit-rates}
\end{figure}

The qutrit result should be interpreted as coexistence, not as a monotonic enhancement by magic. Along the family, changes in purity, entropy, and nonstabilizerness are inseparable. What is established is that nonstabilizer environments can support both a strictly positive classical rate and a strictly positive quantum rate on open parameter regions.

\section{Discussion}

The mean state is an entropy ceiling determined by exact stabilizer constraints. In contrast to the mean-state entropy, which depends only on the exact stabilizer skeleton, the hierarchy retains information about how different compatible stabilizer resolutions reveal the residual degrees of freedom. Our main structural result resolves the route to that ceiling into finite stabilizer-measurement steps. Each step is simultaneously an entropy increment, a relative entropy, and a loss of coherence with respect to one additional compatible stabilizer observable. The optimized profile \(\mathfrak{s}_k(\sigma)\) records the least such loss at each measurement depth.

The communication consequence separates into two related layers. The spectral-transfer argument is stronger than the compatible hierarchy alone: it applies to every complete stabilizer measurement of the environment and yields the optimized stabilizer-dephasing lower bound. The mean-state hierarchy enters by selecting the compatible measurements determined by the exact stabilizer skeleton. This canonical restriction gives the bound $n\log d-\mathfrak{s}_n(\sigma)$, explains why it refines the mean-state bound, and connects the communication problem back to finite-resolution entropy geometry.

Two questions remain central. First, one would like to characterize all environments for which a stabilizer input minimizes the output entropy. Second, even when the single-letter stabilizer-dephasing formula holds, additivity must be established separately. The qutrit family provides an analytic testing ground: its one-shot classical formula is exact, its quantum-capacity lower bound has a sharp symmetry pattern, and its midpoint is unitarily self-complementary.

The hierarchy also suggests an experimental route. Each maximal endpoint is an ordinary stabilizer measurement distribution, while intermediate values correspond to partial compatible measurements. Estimating the profile therefore requires only stabilizer observables, although optimizing over all isotropic subspaces may become combinatorial for large \(n\). Developing efficient restricted profiles and finite-sample guarantees is a natural next step.

\begin{acknowledgments}
We thank Hai Wang, Wentao Qi, and Sijie Luo for discussions. This project is supported by the National Natural Science Foundation of China (Grants No. 12201555, 12050410232, and 12526648) and the Natural Science Foundation of Hunan Province (Grants No. 2025JJ50050 and 2025JJ60025).
\end{acknowledgments}

\section*{Author contributions}

All authors contributed to the research, discussed the results, wrote the manuscript, and approved the final version.


\section*{Competing interests}
The authors declare no competing interests.

\section*{Data and code availability}
No datasets were generated or analyzed. The numerical curves in Fig.~\ref{fig:qutrit-rates} are evaluations of the closed-form expressions in Eqs.~\eqref{eq:qutrit-holevo}, \eqref{eq:lambda-p}, and \eqref{eq:ic0}; plotting code can be supplied with the source files.


\appendix

\section{Qutrit entropy lemmas}
\label{app:qutrit-entropy}

\begin{lemma}[Entropy increases with determinant at fixed qutrit purity]
Let \(\bm\lambda=(\lambda_1,\lambda_2,\lambda_3)\) be a qutrit probability spectrum with fixed trace and fixed purity. On every nondegenerate interior branch, its Shannon entropy is strictly increasing as a function of \(e_3=\lambda_1\lambda_2\lambda_3\). The statement extends to degenerate spectra by continuity.
\end{lemma}

\begin{proof}
Fixed trace and purity fix the first two elementary symmetric polynomials,
\begin{equation*}
 e_1=1,\qquad e_2=\frac{1-\sum_i\lambda_i^2}{2}.
\end{equation*}
The eigenvalues are the roots of
\begin{equation*}
 f(x)=x^3-x^2+e_2x-e_3.
\end{equation*}
For distinct positive roots, implicit differentiation gives
\begin{equation*}
 \frac{\dd\lambda_i}{\dd e_3}=\frac{1}{f'(\lambda_i)}.
\end{equation*}
Since \(\sum_i\dd\lambda_i/\dd e_3=0\),
\begin{align}
 \frac{\dd H(\bm\lambda)}{\dd e_3}
 &=-\frac1{\ln2}\sum_i\frac{\ln\lambda_i}{f'(\lambda_i)}\notag\\
 &=-\frac1{\ln2}[\ln x]_{\lambda_1,\lambda_2,\lambda_3}>0,
 \label{eq:divided-difference}
\end{align}
where \([f]_{a,b,c}\) denotes the second divided difference of \(f\) at \(a,b,c\). The divided difference \([\ln x]_{\lambda_1,\lambda_2,\lambda_3}\) is strictly  negative because \(\ln x\) is strictly concave on \((0,\infty)\). Continuity handles repeated roots and boundary limits.
\end{proof}

\begin{lemma}[Minimum qutrit entropy under a purity upper bound]
Let \(r\in[1/2,1]\). If a qutrit state \(\omega\) satisfies
\begin{equation*}
 \Tr\omega^2\le r^2+(1-r)^2,
\end{equation*}
then
\begin{equation}
 S(\omega)\ge\hbin(r).
 \label{eq:qutrit-purity-entropy}
\end{equation}
Equality is attained by the rank-two spectrum \((r,1-r,0)\).
\end{lemma}

\begin{proof}
If \(\Tr\omega^2\le1/2\), the collision-entropy bound gives
\(S(\omega)\ge-\log\Tr\omega^2\ge1\ge\hbin(r)\).

Now fix a purity \(P\in(1/2,r^2+(1-r)^2]\). At fixed trace and purity, the preceding lemma shows that entropy is minimized by the smallest physically allowed determinant, namely zero. The minimizing spectrum is therefore \((r_P,1-r_P,0)\), where
\begin{equation*}
 r_P=\frac{1+\sqrt{2P-1}}{2}.
\end{equation*}
Since \(P\le r^2+(1-r)^2\), one has \(r_P\le r\). Binary entropy decreases on \([1/2,1]\), so \(\hbin(r_P)\ge\hbin(r)\).

\end{proof}

\begin{lemma}[Entropy of the noncomputational qutrit stabilizer measurements]
For \(p\in[0,1]\), \(q=1-p\), and
\begin{equation*}
 r_k(\theta)=\frac{1+2\sqrt{pq}\cos(\theta+2\pi k/3)}3,
\end{equation*}
one has \(H(r_0,r_1,r_2)\ge\hbin(p)\).
\end{lemma}

\begin{proof}
The trigonometric identities \(\sum_k\cos(\theta+2\pi k/3)=0\) and \(\sum_k\cos^2(\theta+2\pi k/3)=3/2\) give $\sum_{k=0}^2r_k^2=\frac{1+2pq}{3}$.
Thus
\begin{equation*}
 H(r_0,r_1,r_2)\ge\log3-\log(1+2pq).
\end{equation*}
It remains to show \(F(p):=\log[3/(1+2p(1-p))]-\hbin(p)\ge0\). By symmetry it suffices to take \(0\le p\le1/2\). In natural logarithms,
\begin{equation*}
 (\ln2)F'(p)
 =-\frac{2(1-2p)}{1+2p(1-p)}-\ln\frac{1-p}{p}<0
\end{equation*}
for \(0<p<1/2\), while \(F(1/2)=0\). Hence \(F(p)\ge0\).
\end{proof}

\section{Self-complementarity at the qutrit midpoint}

For completeness, let \(\rho=(\rho_{jk})_{j,k=0}^2\). A direct calculation from the convolutional unitary \(U_{G_H}\) with environment \(\sigma_{1/2}\) gives matrices \(A=\Lambda_{\sigma_{1/2}}(\rho)\) and \(B=\Lambda^c_{\sigma_{1/2}}(\rho)\) satisfying
\[
B=XAX^\dagger,
\]
where \(X|j\rangle=|j+1\rangle\). This proves Eq.~\eqref{eq:self-complementary} for arbitrary, including mixed, input states.


\section{Symplectic similitudes induced by the convolution matrix}
\label{app:symplectic-similitudes}

We record the elementary symplectic calculation used in the proof of
Theorem~\ref{thm:holevo-bound}.

\begin{lemma}
\label{lem:AGBG-similitudes}
Let
\[
A_G(p,q)=(N_Gg_{11}p,g_{00}q),B_G(p,q)=(-N_Gg_{10}p,g_{01}q),
\]
where \(N_G=(\det G)^{-1}\). If \(G\) is positive, then \(A_G\) and \(B_G\)
are invertible symplectic similitudes on \(V^n\). More precisely, for all
\(x,y\in V^n\),
\[
\langle A_Gx,A_Gy\rangle_s
=
N_G g_{00}g_{11}\,\langle x,y\rangle_s,
\]
and
\[
\langle B_Gx,B_Gy\rangle_s
=
-\,N_G g_{01}g_{10}\,\langle x,y\rangle_s .
\]
In particular, \(A_G\) and \(B_G\) map isotropic subspaces to isotropic
subspaces and maximal isotropic subspaces to maximal isotropic subspaces.
\end{lemma}

\begin{proof}
Write \(x=(p,q)\) and \(y=(p',q')\). Then
\[
\begin{aligned}
\langle A_Gx,A_Gy\rangle_s
&=
(N_Gg_{11}p)\cdot(g_{00}q')
-
(N_Gg_{11}p')\cdot(g_{00}q)  \\
&=
N_Gg_{00}g_{11}
\bigl(p\cdot q'-p'\cdot q\bigr)  \\
&=
N_Gg_{00}g_{11}\,\langle x,y\rangle_s .
\end{aligned}
\]
Similarly,
\[
\begin{aligned}
\langle B_Gx,B_Gy\rangle_s
&=
(-N_Gg_{10}p)\cdot(g_{01}q')
-
(-N_Gg_{10}p')\cdot(g_{01}q)  \\
&=
-\,N_Gg_{01}g_{10}
\bigl(p\cdot q'-p'\cdot q\bigr) \\
&=
-\,N_Gg_{01}g_{10}\,\langle x,y\rangle_s .
\end{aligned}
\]
If \(G\) is positive, all four entries \(g_{00},g_{01},g_{10},g_{11}\)
are nonzero, so the coordinate formulas show that \(A_G\) and \(B_G\)
are invertible. The displayed identities then imply that they preserve
isotropic subspaces. Since invertible linear maps preserve dimension,
maximal isotropic subspaces are mapped to maximal isotropic subspaces.
\end{proof}

\end{document}